\pdfoutput 1

\documentclass[11pt]{article}

\usepackage{amsfonts}
\usepackage{amsmath, amsthm, amssymb}
\usepackage{algorithm,algorithmic} 
\usepackage{graphicx}    

\usepackage{fullpage}

\newcommand{\opt}{\mbox{opt}}
\newcommand{\MCPDT}{\mbox{MCPDT}}
\newcommand{\MCIP}{\mbox{MCIP}}

\newtheorem{lemma}{Lemma}[section]
\newtheorem{proposition}{Proposition}[section]
\newtheorem{theorem}{Theorem}[section]
\newtheorem{cor}{Corollary}[section]
\newtheorem{definition}{Definition}[section]
\newtheorem{fact}{Fact}[section]
\newtheorem*{remark}{Remark}

\renewcommand{\algorithmicrequire}{\textbf{input:}} 
\renewcommand{\algorithmicensure}{\textbf{output:}}
\renewcommand{\algorithmicforall}{\textbf{foreach}}

\title{Comparing Pedigree Graphs}
\date{}
\author{
  Bonnie Kirkpatrick\footnote{B.K. is with the Electrical Engineering and Computer Sciences, University of California, Berkeley, and 
International Computer Science Institute, Berkeley, {\tt bbkirk@eecs.berkeley.edu}.},
  Yakir Reshef\footnote{Y.R. is with the Weizmann Institute of Science, Rehovot, Israel, \tt{yakir.reshef@weizmann.ac.il}.},
  Hilary Finucane\footnote{H.F. is with the Weizmann Institute of Science, Rehovot, Israel, \tt{hilary.finucane@weizmann.ac.il}.}, \\
  Haitao Jiang\footnote{H.J. is with the School of Computer Science and Technology, Shandong University, China., \tt{htjiang@cs.montana.edu}.},
  Binhai Zhu\footnote{B.Z. is with the Department of Computer Science, Montana State University, Bozeman, MT 59717, USA, \tt{bhz@cs.montana.edu}.},
  Richard M. Karp\footnote{R.M.K. is with the Electrical Engineering and Computer Sciences, University of California, Berkeley, and International Computer Science Institute, Berkeley, \tt{karp@cs.berkeley.edu}.}}

\begin{document}         

\maketitle
\vspace{-0.75cm}


\begin{abstract}
Pedigree graphs, or family trees, are typically constructed by an expensive process of examining genealogical records to determine which pairs of individuals are parent and child. New methods to automate this process take as input genetic data from a set of extant individuals and reconstruct ancestral individuals.  There is a great need to evaluate the quality of these methods by comparing the estimated pedigree to the true pedigree.

In this paper, we consider two main pedigree comparison problems. The first is the pedigree isomorphism problem, for which we present a linear-time algorithm for leaf-labeled pedigrees. The second is the pedigree edit distance problem, for which we present 1) several algorithms that are fast and exact in various special cases, and 2) a general, randomized heuristic algorithm.

In the negative direction, we first prove that the pedigree isomorphism problem is as hard as the general graph isomorphism problem, and that the sub-pedigree isomorphism problem is NP-hard. We then show that the pedigree edit distance problem is APX-hard in general and NP-hard on leaf-labeled pedigrees.

We use simulated pedigrees to compare our edit-distance algorithms to each other as well as to a branch-and-bound algorithm that always finds an optimal solution.
\end{abstract}


\newpage
\section{Introduction}

Pedigrees, or family trees, are of interest in a variety of fields. They are interesting to geneticists due to the accuracy with which recombinations can be inferred~\cite{Coop2008} and with which disease loci can be mapped~\cite{Ng2009,Ng2010}.  Likelihood calculations, i.e.~calculations of the probability of observing the data inherited in a given pedigree, are of great interest for mapping disease loci. Pedigrees are objects of interest in computer science due to their close connection with machine learning methods~\cite{Lauritzen2003,Fishelson2005}.  Many calculations on pedigree graphs are hard~\cite{Piccolboni2003,Li2003,Kirkpatrick2010b}, but notable attempts have been made to improve the speed of these calculations~\cite{Browning2002,Geiger2009,Li2010}.

Reconstructing pedigrees is thus an interesting but difficult problem.  Genealogical methods for reconstructing pedigrees can involve multiple sources with contradictory or missing information~\cite{Anderson2006,Simmons2004,Sun2002, McPeek2000, Boehnke1997}. Due to the error-prone nature of genealogical pedigree reconstruction and the unavailability of genealogical data for some animals, the pedigree reconstruction problem was introduced by Thompson~\cite{thompson1985} as follows: given genetic data for a set of extant individuals, reconstruct relationships between those individuals that may involve unobserved ancestors.  State-of-the-art practical methods include~\cite{thompson1985,Stankovich2005,Berger-Wolf2007,Brown2010,Kirkpatrick2011c} and theoretical work includes~\cite{Thatte2008,Thatte2006}.

Evaluating reconstruction methods requires inferring a pedigree on an instance for which the true pedigree is known and comparing the inferred and known pedigrees. Both the estimated pedigree and the true pedigree will have the same set of extant individuals---i.e. the individuals having genetic data---but may have different inferred ancestors.  Thus, to compare these two pedigrees, we must compare their topology in a fashion that respects the labels of the extant individuals.

Existing methods of comparing pedigrees are insufficient.  For example, phylogenetic tree comparison methods can only be used to compare tree-like pedigrees, but pedigrees can take more general forms.  As another example, ~\cite{Kirkpatrick2011c} evaluates accuracy using the kinship coefficient of all pairs of individuals, where the kinship coefficient of two individuals at a single locus is the average number of alleles inherited from the same ancestor.  This is a poor accuracy measure, since the kinship coefficient is not identifiable.  For example, two half-siblings have the same kinship coefficient as an uncle and nephew.  A recent result demonstrates non-identifiability of larger pedigrees~\cite{Pinto2010}.  Furthermore, the pedigree likelihood is not an acceptable pedigree comparison method both because it requires an exponential-time algorithm and because the non-identifiability of the kinship coefficient may well imply the non-identifiability of the likelihood.  While pedigree isomorphisms were discussed very briefly by Steel and Hein~\cite{Steel2006} in the context of pedigree reconstruction, they did not discuss the pedigree isomorphism problem and its algorithms.  Finally, brute force methods of pedigree comparison are not sufficient, as can be seen in our own brute-force comparison where the simulation was limited to pedigrees of fourteen individuals.  This is in contrast to pedigree data sets which include thousands of individuals~\cite{Abney2002,Sutter2007}.  Other biologists are collecting data from hundreds of individuals from large families---i.e. \~60 individuals per family for salmon~\cite{Herbinger1999,Almudevar1999}---where brute-force and likelihood methods fall short due to exponential running times.

Two natural formulations of the pedigree comparison problem are discussed in this paper: pedigree isomorphism and pedigree edit distance.  Two pedigrees are isomorphic if there exists a graph isomorphism which respects the genders of all individuals and the identities of the individuals for which we have genetic data; the pedigree isomorphism problem is to determine whether two pedigrees are isomorphic. The more difficult pedigree edit distance problem is to determine how many edge insertions and deletions are needed for two pedigrees to become isomorphic.

In this paper, we formalize the isomorphism and edit distance problems, provide useful algorithms for certain instances of these problems, and give four hardness results.  The algorithms we present include a fast algorithm for leaf-labeled pedigree isomorphism and a polynomial-time dynamic programming algorithm for the edit distance of sufficiently similar pedigrees.  For general pedigrees, we present a heuristic algorithm for the edit distance.

Our first hardness result is that pedigree isomorphism is as hard as general graph isomorphism, making it GI-hard. The reduction we use also shows that sub-pedigree isomorphism is as hard as sub-graph isomorphism, making it NP-hard. The third and fourth hardness results show APX-hardness for the edit distance problem in general and NP-hardness on pedigrees whose leaves are all labeled. 

\section{Preliminaries}
\label{sec:preliminaries}

A \em pedigree \em $\mathcal{P} = (P,s,X,\ell)$ consists of a {\bf pedigree graph} $P = (I(P),E(P))$ with vertices $I(P)$ and edges $E(P)$, a {\bf gender function} $s : I(P) \rightarrow \{m,f \}$, a set $X \subseteq I(P)$ of {\bf labeled individuals}, and an injective {\bf labeling} $\ell : X \rightarrow \mathbb{N}$, such that:
\vspace{-0.25cm}
\begin{enumerate}
\item $P$ is directed and acyclic.
\vspace{-0.25cm}
\item For all $v \in V$, the in-degree of $v$ is either two or zero.
\vspace{-0.25cm}
\item If $(a,v), (b,v) \in E$, then $s(a) \neq s(b)$.
\end{enumerate}
\vspace{-0.25cm}
$P$ is called the {\bf pedigree graph} of $\mathcal{P}$. Vertices of $P$ are called the {\bf individuals} of $\mathcal{P}$. Individuals with in-degree zero are {\bf founders} while individuals with in-degree two are {\bf non-founders}. Individuals with out-degree zero are called {\bf leaf individuals}.  For an individual $x \in X$, $\ell(x)$ is the {\bf label} assigned to $x$. We will sometimes write $\mathcal{P} = (P,s)$ to mean $\mathcal{P} = (P,\emptyset,s,\ell)$ with trivial $\ell$.  When $X=I(P)$, we will say that ${\cal P}$ is {\bf fully labeled}.  If the reference pedigree is clear, we may write $I$ to refer to $I(P)$.   Figure~\ref{fig:example} depicts an example of a fully labeled pedigree. 

\begin{figure}
\begin{center}
\includegraphics[scale=1]{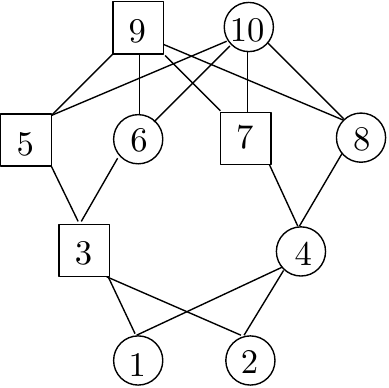}
\caption{{\bf An Example Pedigree.}  
This pedigree has edges implicitly directed downward from parent to child and shows two founding grand-parents, their four children, two inbred grand-children and two inbred great-grand-children. Each edge represents the transmission from parent $i$ to offspring $j$ of a single (possibly recombinant) copy of each chromosome. }
\label{fig:example}
\end{center}
\end{figure}

The labeled individuals $X \subseteq I$ typically have an available DNA sample on which genotyping or sequencing can be performed.  While it may be  algorithmically convenient to assume that we have samples for all individuals, $X = I$, this assumption is quite impractical.  Very often, there are individuals represented in the pedigree who are deceased and for whom samples are unavailable.  There are circumstances in which there will be no labeled individuals, i.e., $X = \emptyset$. Then we would want to rely on the genealogical structure alone in determining whether the same individuals appear in the two pedigree graphs.

A pedigree $\mathcal{P} = (P,s,X,\ell)$ is {\bf leaf-labeled} if $X$ contains exactly the leaf-individuals of $\mathcal{P}$.
Very often this is the case since, for likelihood calculations, one is only interested in the individuals for whom identifying data exist. 

Two pedigrees $\mathcal{P} = (P,s,X,\ell)$ and $\mathcal{P'} = (P',s',X',\ell')$ are {\bf compatibly leaf-labeled} if both are leaf-labeled and for every leaf individual $v \in I(P)$, there is a leaf individual $v' \in I(P')$ such that $\ell(v) = \ell'(v')$ and vice versa.  Given two leaf-labeled pedigree graphs $P$ and $Q$, we can always obtain new pedigree graphs $P'$ and $Q'$ that are compatibly leaf-labeled.  This is done by performing a depth-first search starting at each compatibly leaf-labeled individual $v$, following parent-child edges backward in time, and then pruning those individuals which were not visited by the depth-first search.  This is interesting, because this allows us to detect an isomorphic subgraph which is compatibly leaf-labeled.

A pedigree gives rise to sub-pedigrees in the following way. 
If two pedigrees $\mathcal{P} = (P,s,X,\ell)$ and $\mathcal{P'} = (P',s',X',\ell')$ satisfy:
\begin{enumerate}
\vspace{-0.25cm}
\item $I(P') \subseteq I(P)$ and $P'$ is the subgraph of $P$ induced by $I(P')$,
\vspace{-0.25cm}
\item  $X' = X \cap I(P')$,
\vspace{-0.25cm}
\item $\ell' = \ell|_{I(P')}$, and 
\vspace{-0.25cm}
\item $s' = s|_{I(P')}$,
\vspace{-0.25cm}
\end{enumerate}
then we say that $\mathcal{P'}$ is a {\bf sub-pedigree} of $\mathcal{P}$, and write $\mathcal{P'} \subseteq \mathcal{P}$.

Given a set $A \subseteq I(P)$, we will write $\mathcal{P}|_A$ to denote the minimal sub-pedigree of $\mathcal{P}$ containing the vertices in $A$. Likewise, if $A \subseteq E(P)$, we write $\mathcal{P}|_A$ to denote the minimal sub-pedigree of $\mathcal{P}$ containing the edges in $A$.

There are two additional types of pedigrees in which we are interested: monogamous and generational pedigrees.
A pedigree $\mathcal{P} = (P,s,X,\ell)$  is {\bf monogamous} if all the individuals are monogamous.  An individual $v \in I(P)$ is {\bf monogamous} if the number of individuals $v' \neq v$ such that $(v,u),(v',u) \in E(P)$ for some $u \in I(P)$ is at most 1.

A pedigree $\mathcal{P} = (P,s,X,\ell)$ is {\bf generational} if there exists a function $G : I(P) \rightarrow \mathbb{N}$ such that the following conditions hold: 
\begin{enumerate}
\vspace{-0.25cm}
\item $G(v) = 1$ for some $v \in I(P)$ where $v$ has in-degree zero, and 
\vspace{-0.25cm}
\item If $(u,v) \in E(P)$, then $G(v) = G(u) + 1$.
\vspace{-0.25cm}
\end{enumerate}
The number $G(v)$ is called the {\bf generation} of $v$, and $G$ is called the {\bf generation map} of $\mathcal{P}$. 
Whenever $\mathcal{P}$ is a generational pedigree with pedigree graph $P$, we will use $I_g(P)$ to denote the individuals of $\mathcal{P}$ whose generation is $g$. 
We will say that a pedigree is {\bf connected} if its pedigree graph is weakly connected. It is easy to see that a connected, generational pedigree has a unique generation map. The maximal value of this map is the number of generations of the pedigree.


\subsection{Pedigree Isomorphism}
\label{subsec:isomorphismdefs}

We now define the notion of an isomorphism between pedigrees. To do so, we first present the more general idea of a matching between pedigrees.

\begin{definition}
Given two pedigrees $\mathcal{P} = (P,s,X,\ell)$ and $\mathcal{P'} = (P',s',X',\ell')$, and a set $Y\subseteq I(P)$, an injective map $M : Y \rightarrow I(P')$ is a {\bf pedigree matching} between $\mathcal{P}$ and $\mathcal{P'}$ if it satisfies the following conditions.
\begin{enumerate}
\vspace{-0.25cm}
\item For every $v \in Y$, $s(v) = s'(M(v))$.
\vspace{-0.25cm}
\item For all $n \in \ell(X) \cap \ell'(X')$, $Y$ contains $\ell^{-1}(n)$ and $M(\ell^{-1}(n)) = (\ell')^{-1}(n)$.
\vspace{-0.25cm}
\end{enumerate}
\end{definition}
Informally, the second condition states that $M$ should respect the labellings $\ell$ and $\ell'$ in the sense that if $\ell$ and $\ell'$ give the same label to two vertices $u$ and $v$ respectively, then $M$ should match $u$ up to $v$.

We can now characterize a pedigree isomorphism as a matching that is also a graph isomorphism between the two pedigree graphs.

\begin{definition}
Given two pedigrees $\mathcal{P} = (P,s,X,\ell)$ and $\mathcal{P'} = (P',s',X',\ell')$, a bijection $\phi : I(P) \rightarrow I(P')$ is a {\bf pedigree isomorphism} between $\mathcal{P}$ and $\mathcal{P'}$ if:
\begin{enumerate}
\vspace{-0.25cm}
\item $\phi$ is a pedigree matching between $\mathcal{P}$ and $\mathcal{P'}$, and 
\vspace{-0.25cm}
\item $(u,v) \in E(P)$ if and only if $(\phi(u), \phi(v)) \in E(P')$
\vspace{-0.25cm}
\end{enumerate}
\end{definition}

\vspace{-0.25cm}
\paragraph{The Pedigree Isomorphism Problem.} Given two pedigrees $\mathcal{P}$ and $\mathcal{P'}$, does there exist a pedigree isomorphism between them?
\vspace{-0.50cm}
\paragraph{The Compatibly Leaf-Labeled Pedigree Isomorphism Problem.} Given two compatibly leaf-labeled pedigrees $\mathcal{P}$ and $\mathcal{P'}$, does there exist a pedigree isomorphism between them?
\vspace{-0.50cm}
\paragraph{The Sub-Pedigree Isomorphism Problem.} Given two pedigrees $\mathcal{P}$ and $\mathcal{P'}$, does there exist a pedigree isomorphism between $\mathcal{P}$ and some sub-pedigree of $\mathcal{P'}$?

In this paper, we will show that the compatibly leaf-labeled pedigree isomorphism problem can be solved in linear time. We will also show that the pedigree isomorphism problem is as hard as the general graph isomorphism problem, and that the sub-pedigree isomorphism problem is NP-hard.


\subsection{Edit Distance}
\label{subsec:editdistancedefs}
We are interested not only in determining whether two pedigrees are isomorphic, but also in how close they are to being isomorphic. Such a measure of distance between pedigrees would be useful for evaluating pedigree reconstruction methods: we could take a known pedigree, extract a subset of its individuals, reconstruct a pedigree from those individuals using our method of choice, and check the distance between the reconstructed pedigree and the true pedigree.

Informally, given two arbitrary pedigrees ${\cal P} = (P, s, X, \ell)$ and ${\cal P'} = (P', s', X', \ell')$, we want to find the minimum number of edge additions/deletions required to convert ${\cal P}$ into ${\cal P'}$. We call this the {\em edit distance} between $\mathcal{P}$ and $\mathcal{P'}$.  Notice that it is not necessary that $|I(P)| = |I(P')|$, because addition/deletion of edge-less vertices will be free.

Formally, we can define edit distance in terms of matchings between pedigrees. To do this, we need to measure how close a matching is to being a pedigree isomorphism. This is done by looking at the set of edges that are \em well matched \em by the matching.
\begin{definition}
Given two pedigrees $\mathcal{P} = (P,s,X,\ell)$ and $\mathcal{P'} = (P',s',X',\ell')$ and a matching $M : Y \rightarrow I(P')$ between them, the set $W_M$ of edges {\bf well-matched} by $M$ is defined by
\[ W_M = \{ (u,v) \in E(P) : u,v, \in Y, (M(u),M(v)) \in E(P') \} .\]
\end{definition}

Notice that the subgraph of $P$ induced by the edges in $W_M$ is isomorphic to the subgraph of $P'$ induced by the edges in $M(W_M)$. Therefore, $M$ implicitly defines an edit path from $P$ to $P'$: delete all the edges in $E(P) - W_M$, then add all the edges in $E(P') - M(W_M)$.  Moreover, the shortest edit path from $P$ to $P'$ must consist only of removing edges from $P$ and adding edges from $P'$, and so can be obtained in this way from a matching (up to the order in which edges are added and removed, which does not affect the length of the edit path). With this in mind, we define the match distance incurred by $M$ to be length of the corresponding edit path:

\begin{definition}
\label{def:matchdistance}
Given two pedigrees $\mathcal{P} = (P,s,X,\ell)$ and $\mathcal{P'} = (P',s',X',\ell')$ and a matching $M : Y \rightarrow I(P')$ between them, the {\bf match distance} of $M$ is
\[ d(M) = d_P(M) + d_{P'}(M) \]
where 
$d_{P}(M) = |E(P) - W_M|$ and $d_{P'}(M) = |E(P') - M(W_M)| $
\end{definition}

Now the edit distance is the length of the shortest edit path; in other words, the smallest possible match distance between the two pedigrees.

\begin{definition}
\label{def:editdistance}
Given two pedigrees $\mathcal{P} = (P,s,X,\ell)$ and $\mathcal{P'} = (P',s',X',\ell')$, the {\bf edit distance} between $\mathcal{P}$ and $\mathcal{P'}$ is
\[ D_{\mathcal{P}, \mathcal{P'}} = \min_{M}{d(M)}. \]
\end{definition}

The edit distance between two pedigrees can be calculated by finding a matching $M$ with a maximum number of well-matched edges. The following proposition, which we state without proof, formalizes this.

 \begin{proposition}
 \label{prop:dualToEditDistance}
 Given two pedigrees $\mathcal{P}$ and $\mathcal{P'}$, a matching $M$ between $\mathcal{P}$ and $\mathcal{P'}$ for which $|W_M|$ is maximized satisfies $D_{\mathcal{P}, \mathcal{P}'} = d(M)$.
 \end{proposition}

\vspace{-0.50cm}
\paragraph{The Pedigree Edit Distance Problem.}  For pedigrees ${\cal P}$ and ${\cal P}'$, find a matching $M$ between $\mathcal{P}$ and $\mathcal{P}'$ such that $d(M) = D_{\mathcal{P},\mathcal{P}'}$.

Notice that a pedigree isomorphism, if it exists, has a match distance of 0.  
In this paper, we will give efficient algorithms for a few different restrictions of the pedigree edit distance problem.
The four main problems and their hardness, as established in this paper, are shown in the table.

\begin{center}
  \begin{tabular}{| l | l | l | }
    \hline
                                   & {\bf Isomorphism} & {\bf Edit Distance} \\ \hline
    {\bf Compatibly Leaf-Labeled } & Linear alg. & NP-hard \\ \hline
    {\bf Not Labeled           }   & GI-hard     & APX-hard \\
    \hline
  \end{tabular}
\end{center}


\section{A Linear-Time Algorithm for the Compatibly Leaf-Labeled Pedigree Isomorphism Problem}
\label{sec:isomorphismalg}

In this section, we introduce a linear-time algorithm for the compatibly leaf-labeled pedigree isomorphism problem. In particular, we will establish that if all leaves are labeled then there exists a total order on the individuals that is easy to calculate.  The total orders, calculated on both pedigrees, are such that if an isomorphism exists between two compatibly leaf-labeled pedigrees, it can be easily found from these total orders.

\begin{proposition} There exists a linear-time algorithm for the compatibly leaf-labeled pedigree isomorphism problem. \end{proposition}

\begin{proof}

We define gender topological sort as follows. Recall that the traditional topological sort algorithm does a depth first search (DFS), and when each node is finished being visited, it gets pushed into the ordered list. Gender topological sort consists of a DFS of the ancestor tree of each leaf, where the female parent is always visited before the male parent (i.e.~the DFS visits in the opposite direct of the directed edge).  From a single leaf, this rule determines the order in which the ancestral nodes of the leaf are visited.  Now, we simply use the leaf labeling to sort the leaves, and we begin our DFS from each leaf in sorted order.  This algorithm finds a total order on the nodes of the pedigree that is fully determined by the topology, gender, and leaf labels of the pedigree.

Let the binary relation $<$ be the total order found by the above gender topological sort.  Let $n = |I(P)| = |I(P')|$. (If $|I(P)| \neq |I(P')|$ then there is no isomorphism.)  The nodes of pedigree ${\cal P}$ are ordered via gender topological sort so that $p_1 < p_2 < ... < p_n$ where $p_i \in I(P)$.  Similarly, the nodes of pedigree ${\cal P'}$ are ordered so that $p'_1 < p'_2 < ... < p'_n$ where $p'_i \in I(P')$.  Because isomorphism preserves topology, gender, and labels, it must also preserve this total order. Thus, if an isomorphism exists between ${\cal P}$ and ${\cal P}$', it must be $\phi$ defined as $\phi(p_i) = p'_i$. So to check whether ${\cal P}$ and ${\cal P}'$ are isomorphic, it suffices to compute the gender topological sort for each pedigree, and then check whether $\phi$ is an isomorphism.

The running time of this isomorphism algorithm is linear, since the gender topological sort is linear, and after obtaining $\phi$, checking that the genders, edges, and labels are preserved also requires linear time.
\end{proof}

Notice that a small modification of this algorithm can find leaf-labeled subgraph isomorphisms.


\section{Algorithms for Computing the Edit Distance}
\label{sec:editdistancealgs}

We show later in this paper that, even for monogamous pedigrees, there is no polynomial-time approximation scheme for the edit distance problem unless $P=NP$. In this section, we show that if we restrict the scope of the problem, efficient algorithms are possible. Specifically, we will give exact, efficient algorithms for the following two restricted cases of the pedigree edit distance problem. In this section we assume that the pedigrees are connected, but this condition can be easily removed.

\begin{enumerate}
\item The case in which the two pedigrees are generational, compatibly leaf-labeled, and both have two generations.
\item The case in which the two pedigrees are generational, compatibly leaf-labeled, and ``sufficiently similar''. By sufficiently similar we mean that there exists an optimal matching between the pedigrees that preserves generations and that for any two consecutive generations $i$ and $i+1$, the distance between the two sub-pedigrees obtained by restricting both pedigrees to generations $i$ and $i+1$ is small. Our algorithm for this case has the advantage that its run-time improves as the pedigrees become more similar.
\end{enumerate}

We then give a randomized heuristic that appears to work well in the general case and is based on an alternate characterization of pedigrees in terms of lists of descendants rather than parent-child relationships, as well as a faster heuristic for the second case listed above.

In the rest of this section, we will denote the two pedigrees under consideration by $\mathcal{P} = (P,s,X,\ell)$ and $\mathcal{P'} = (P',s',X',\ell')$.


\subsection{Exact Algorithm for Two-Generation, Compatibly Leaf-Labeled Pedigrees}
\label{subsec:leaflabelledisomorphism}

When $\mathcal{P}$ and $\mathcal{P'}$ are connected, generational, compatibly leaf-labeled, and have two generations each, the edit distance between them can be calculated in polynomial time. We find this distance by constructing two maximum-weight bipartite matching instances, one for each gender, whose solutions give us an optimal matching between $\mathcal{P}$ and $\mathcal{P'}$.  In doing this, we are maximizing the number of well-matched edges, which is equivalent to minimizing the distance.

Recall that $I_g(P)$ is the set of individuals in the $g$th generation of $P$. Our assumption that both pedigrees are compatibly leaf-labeled determines the matching $M$ on $I_2(P)$: map each $v \in I_2(P)$ to $(\ell')^{-1}(\ell(v)) \in I_2(P')$. In addition, we may assume without loss of generality that no individuals in the oldest generation of $P$ and $P'$ are labeled. This is because the labels that $I_1(P)$ and $I_1(P')$ share determine the matching on those vertices, and the labels not shared are irrelevant to the edit distance.

It is left to extend $M$ optimally to $I_1(P)$.  To do this, we first define $I_1^f(P)$ and $I_1^m(P)$ to be the females and males of $I_1(P)$ respectively, and define $I_1^f(P')$ and $I_1^m(P')$ analogously. Now, for each gender $s$, we construct a complete bipartite graph $G^s$ with left vertices $I_1^s(P)$ and right vertices $I_1^s(P')$ where the weight assigned to an edge $(u,v)$ is the number of children of $u$ who are matched by $M$ to children of $v$. Together with Proposition~\ref{prop:dualToEditDistance}, the following proposition establishes that solving the maximum-weight bipartite matching instances $G^f$ and $G^m$ will yield an optimal matching of $\mathcal{P}$ and $\mathcal{P'}$.

\begin{proposition}
For $s \in \{m,f\}$, let $M^s$ be a perfect matching in $G^s$, and extend the matching $M$ to a matching $\bar{M}$ defined on all of $I(P)$ as follows: for $v \in I_1(P)$, define $M(v)$ to be the vertex matched to $v$ by $M^{s(v)}$. Then the number of edges well-matched by $\bar{M}$ is the sum of the weights of $M^f$ and $M^m$.
\end{proposition}


\subsection{Exact Algorithm for Sufficiently Similar, Generational, Compatibly Leaf-Labeled Pedigrees}
\label{subsec:DP}

Suppose that there exists an optimal matching $M$ between ${\cal P}$ and ${\cal P'}$ that is \em generation preserving \em (i.e., such that the generation of $v$ equals the generation of $M(v)$ for all individuals $v$ matched by $M$). We will also assume for simplicity that ${\cal P}$ and ${\cal P'}$ are each made up of $g$ generations of $m$ males and $m$ females each, though this assumption can be easily removed. If the pedigree graphs $P$ and $P'$ are similar enough, we can use dynamic programming to find an optimal matching between them in time exponential only in the edit distance. Thus, we have a polynomial-time algorithm if, for any two consecutive generations $i$ and $i+1$, the edit distance between the two sub-pedigrees obtained by restricting both pedigrees to generations $i$ and $i+1$ is at most $O(log(n)/g)$, where $n$ is the number of pedigree individuals.

To describe our algorithm, we first introduce some notation. Given some $S \subseteq \{1,2,...,g\}$, we let ${\cal P}|_S$ denote the minimal sub-pedigree of $\mathcal{P}$ that contains $I_i(P)$ for every $i \in S$, and we define $\mathcal{P'}|_S$ analogously. We can then write $\mathcal{M}(S)$ to denote the set of all generation-preserving matchings from ${\cal P}|_S$ to ${\cal P'}|_S$; note that 
\[\mathcal{M}(S) = \times_{i \in S}{\mathcal{M}(\{i\})}\]
where $\mathcal{M}(\{i\})$ is the set of matchings of generation $i$.
Given a generation-preserving matching $M$, let $d_S(M)$ be the match distance incurred by $M$ restricted to be a matching between ${\cal P}|_S$ and $\mathcal{P'}|_S$, and let $B_i(M)$ be the edit distance between ${\cal P}|_{\{i,\ldots,g\}}$ and ${\cal P'}|_{\{i,\ldots,g\}}$ taken only over matchings that agree with $M$ wherever $M$ is defined.

As in the previous section, we assume without loss of generality that \em only \em the leaf individuals of either pedigree are labeled.

The algorithm we present rests on the following relation.
\begin{lemma}
\label{DPrelation}
For every $i \in \{1,\ldots,g-1\}$, and for every $M \in \mathcal{M}(\{i\})$, we have
\begin{equation} 
B_i(M) = \min_{M' \in \mathcal{M}(\{i+1\})} {B_{i+1}(M') + d_{\{i,i+1\}}(M \cup M')}
\end{equation}
where $M \cup M' \in \mathcal{M}(\{i,i+1\})$ denotes the matching that equals $M$ on $\mathcal{P}|_{\{i\}}$ and $M'$ on $\mathcal{P'}|_{\{i+1\}}$.
\end{lemma}
Lemma~\ref{DPrelation} gives rise to a simple dynamic programming algorithm: start with the matching of the $g$th generation determined by the labellings of the leaves, then iteratively work up the pedigree, using the lemma above to find, for every $i$, the values of $B_i(M)$ for every $M \in \mathcal{S}(\{i\})$. The edit distance is then given by
\[\min_{M \in \mathcal{M}(\{1\})}{B_1(M)}.\]

However, the problem with this straightforward algorithm is that because it needs to consider every possible matching $M \in \mathcal{M}(\{i\})$, its run-time is factorial in $m$, the number of males/females in each generation. At each generation, there are $(m!)^2$ possible matchings $M$ to process, and performing the minimization for each matching takes time $O((m!)^2)$. Therefore, the run-time of this algorithm is $O(g(m!)^{4})$.

We can improve the algorithm if we know that the two pedigrees under consideration are sufficiently similar at each generation and so there is no need to consider all matchings for each generation.
Suppose we are promised that an optimal matching $\bar{M}$ between ${\cal P}$ and ${\cal P'}$ has $d_{\{i,i+1\}}(\bar{M}) < k$ for every $1 \leq i < g$. Then in the algorithm above we would only need to process, for each $i$ and each $M \in \mathcal{M}(\{i\})$, the matchings $M' \in \mathcal{M}(\{i+1\})$ such that $d_{\{i,i+1\}}(M \cup M') < k$. This enumeration can be done, for a fixed $M' \in \mathcal{M}(\{i+1\})$ of the previous generation, by recursively matching vertices to gradually define $M \in \mathcal{M}(\{i\})$, all the time avoiding any assignment that would violate the condition $d_{\{i,i+1\}}(M \cup M') < k$.  
The case of a small edit distance is particularly interesting.  Our simulations, see Fig.~\ref{fig:accuracy}, show that when some number $x$ of random changes are made to a pedigree, the edit distance is close to $x$ when $x$ is small but for larger values of $x$, some changes cancel each other out and the edit distance grows more slowly than $x$.  Thus, when the edit distance is small, it corresponds more closely to the actual number of changes made.

How many matchings are considered by this method? The following two lemmas establish an upper bound of $m^{2k}$ on this quantity.

\begin{lemma}
\label{lem:recurrence}
For every fixed $M' \in \mathcal{M}(\{i+1\})$, the number of matchings $M \in \mathcal{M}(\{i\})$ such that $d_{\{i,i+1\}}(M \cup M') < k$ is at most $T(m,k)^2$, where $T$ satisfies the recurrence relation
\[ T(n,c) = T(n-1,c) + (n-1)T(n-1,c-2) \]
with initial conditions $T(1,\cdot) = 1$ and $T(n,0) = T(n,1) = 1$.
\end{lemma}
\begin{proof}
First, suppose that there are only $m$ individuals to match (i.e. that there is only one gender). Initially, there are $m$ individuals to match and $k$ ``cost points'' that can be used in doing so. In the best case, given some $u \in \mathcal{P}|_{\{i\}}$, there is at most one choice for $M(u)$ that does not increase the cost of the matching being built.  This follows from the fact that $u$ has at least one child (otherwise $u$ is a labeled leaf and so $M(u)$ is already determined), who is already matched somewhere by $M'$. Besides this option for $M(u)$, there are at most $m - 1$ other options, each of which will increase the cost of the matching by at least $2$ (since at least one edge will have to be deleted and one edge will have to be added). This establishes the recurrence. The initial conditions follow from the following facts:
\begin{enumerate}
\item When one individual is left to be matched, there is only one possible way to complete the matching being built.
\item When the matching being built already has distance $k$ (i.e. there are $0$ cost points left), there is at best only one way to complete the matching.
\end{enumerate}
This bounds the number of matchings of each gender by $T(m,k)$. Since this process occurs independently for each gender, the number of total matchings is at most $T(m,k)^2$.
\end{proof}

\begin{lemma}\label{DPruntime}
The recurrence $T$ in Lemma~\ref{lem:recurrence} satisfies $T(n,c) = O(n^c)$.
\end{lemma}
\begin{proof}
We proceed by induction on $c$. The initial conditions of $T$ give us our base cases of $c = 0, 1$. The general case follows from bounding the difference between successive values of $T(\cdot, c)$: the recurrence gives us that $T(n,c) - T(n-1,c) = (n-1)T(n-1,c-2)$, which is $n \cdot O(n^{c-2}) = O(n^{c-1})$ by the inductive hypothesis.
\end{proof}

Thus, the number of matchings to be considered at generation $i$ is $m^{2k}$ times the number of matchings to be considered at generation $i+1$. So the run-time of this algorithm is dominated by its last step, in which matchings between the oldest generations of ${\cal P}$ and ${\cal P'}$ are considered and the best one is chosen; the number of these matchings is at most $O(m^{2k(g-1)}) = O(m^{2d})$ where $d$ is the maximum possible distance between the two pedigrees. Thus, if $k$ (the distance between pairs of successive generations) and $g$ (the number of generations) are small, the algorithm can efficiently calculate edit distance.

This algorithm always finds the correct edit distance, when the upper bound $k$ is known. When $k$ is not known, the algorithm can be adapted to a heuristic by guessing a reasonable $k$, and if there is a step with no matching within distance $k$, the algorithm is aborted and the randomized heuristic (described below) is used instead. In Section~\ref{sec:simulations}, we show by comparison to a branch-and-bound algorithm that tries all possible matchings that this heuristic adaptation often finds the correct edit distance when used on randomly generated pedigrees.


\subsection{Heuristic Improvement of Dynamic Programming Algorithm}
\label{subsec:heuristicDP}

We can turn the dynamic programming algorithm from Section~\ref{subsec:DP} into a faster heuristic by enumerating a still smaller set of matchings. For each generation and for a fixed labeling of the previous generation, we can create an instance of the maximum-weight bipartite matching problem. However, instead of solving the problem exactly, we can enumerate its $\gamma$ best solutions and consider those matchings only.  Since this can be done in time $O(\gamma m^3)$ (see~\cite{Chegireddy1987}), this improves the running-time of the algorithm to $O(m^3 \gamma^{g-1})$.


\subsection{Randomized Heuristic for Compatibly Leaf-Labeled Pedigrees}
\label{subsec:randomheuristic}
The randomized matching algorithm for regular pedigrees rests on the idea of viewing pedigrees not as lists of parent-child relationships, but rather as sets of so-called \em descendant splits\em. The descendant split of an individual $u \in I(P)$ is the set of individuals $v \in I(P)$ such that there is a directed path from $u$ to $v$. When a pedigree is fully labeled, its full set of descendant splits uniquely specifies it. For more on descendant splits see~\cite{Kirkpatrick2011b}.

The heuristic calculates the match distance incurred by a matching that is randomly selected as follows: at each generation, among individuals of the same gender, an individual $u \in I(P)$ is matched to $v \in I(P')$ with probability proportional to the number of individuals in the descendant splits of the two individuals which are identically labeled. This can be done in polynomial time by creating, for each generation, an $m \times m$ matrix of individual match probabilities for each gender (where there are $2m$ individuals per generation). The
matches are then drawn from these probability matrices (without replacement of previously matched individuals). This can be done multiple times to increase the chances of finding a `good' matching.

This algorithm is difficult to analyze because different leaves do not always have disjoint paths to the vertex being matched. However, we show in simulations (Section~\ref{sec:simulations}) that it performs reasonably well relative to a branch-and-bound algorithm that considers all possible matchings.


\section{Hardness Results}
\label{sec:hardness}

Having presented algorithms for various restrictions of the pedigree isomorphism and edit distance problems, we now establish the difficulty of solving the general versions of these problems. Specifically, we give hardness results for the pedigree isomorphism problem (GI-Hard), the sub-pedigree isomorphism problem (NP-Hard), the pedigree edit distance problem (APX-Hard), and the compatibly leaf-labeled pedigree edit distance problem (NP-Hard).
The first few proofs require that $X = \emptyset$, while the final proof considers compatibly leaf-labeled pedigrees.


\subsection{The Hardness of Pedigree Isomorphism and Sub-Pedigree Isomorphism}
\label{subsec:isomorphismhardness}
\subsubsection{Pedigree Isomorphism is GI-Hard}
We begin by showing that the pedigree isomorphism problem is as hard as the general graph isomorphism problem. Graph isomorphism is one of very few problems not known to be in P that is also not known to be NP-Hard. Problems that are as hard as graph isomorphism are known as GI-Hard. To show our hardness result, we reduce from bipartite graph isomorphism, which is known to be GI-Hard~\cite{Uehara2005}.

\paragraph{Reduction:} Given a bipartite graph $G = (V_1\cup V_2, E)$, we define a pedigree ${\cal P}(G) = (P,s)$. Intuitively, each vertex $u \in V_1$ is replaced with two vertices, a male denoted $u^m$ and a female denoted $u^f$, and for each edge $(u,v)$ in $E$, the couple corresponding to $u^m$ and $u^f$ have two children, a male denoted $(u,v)^m$ and a female denoted $(u,v)^f$. This gives as many couples corresponding to $v$ as there are edges going into $v$. To encode the fact that all of these couples correspond to edges with the same vertex $v \in V_2$, we have every member of the form $(u,v)^m$ for some $u \in V_1$ mate with every member of the form $(u',v)^f$ for some $u' \in V_1$ to obtain a female child which we denote $(u,u',v)^f$. Formally, this gives us the following
pedigree graph:
\begin{itemize}
\item $I(P) = I_1 \cup I_2 \cup I_3$, where
        \begin{itemize}
        \item $I_1 = \{u^m, u^f : u \in V_1\}$
        \item $I_2 = \{(u,v)^m, (u,v)^f : (u,v)\in E\}$
        \item $I_3 = \{(u,u',v)^f : (u,v),(u',v) \in E\}$ 
        \end{itemize}
\item $E(P) = E_1 \cup E_2$, where
        \begin{itemize}
        \item $E_1 = \{(u^s,(u,v)^{s'}) : (u,v)\in E, s,s' \in \{m,f\}\}$
        \item $F_2 = \{((u,v)^m,(u,u',v)^f), ((u,v)^f,(u',u,v)^f) : (u,v),(u',v) \in E\}$
        \end{itemize}                
\end{itemize}
By construction, we have:  
\begin{lemma} $P$ is a pedigree graph.
\end{lemma}
Next, we show:
\begin{proposition}
\label{prop:GIhard}
Two bipartite graphs $G = (V_1 \cup V_2, E)$ and $G' = (V_1', V_2', E')$ with no isolated vertices are isomorphic if and only if the two pedigrees $\mathcal{P}(G) = (P,s)$ and $\mathcal{P}(G') = (P',s')$ are isomorphic.
\end{proposition}
\begin{proof}
($\Rightarrow$:) Suppose we have a graph isomorphism $\varphi : V_1 \cup V_2 \rightarrow V_1' \cup V_2'$. It is easy to verify that the following definitions give a map $\phi : I(P) \rightarrow I(P')$ that is a pedigree isomorphism.
\begin{itemize}
\item $\phi(u^s) = \varphi(u)^s$ for $u^s \in I_1$
\item $\phi((u,v)^s) = (\varphi(u),\varphi(v))^s$ for $(u,v)^s \in I_2$
\item $\phi((u,u',v)^f) = (\varphi(u),\varphi(u'),\varphi(v))^f$ for $(u,u',v)^f \in I_3$
\end{itemize}

($\Leftarrow$:) We write $I(P) = I_1 \cup I_2 \cup I_3$ and write $I(P') = I_1' \cup I_2' \cup I_3'$. Using this notation, our assumption gives us an injective map $\phi:I_1 \cup I_2 \cup I_3 \rightarrow I_1'\cup I_2' \cup I_3'$ that is a graph isomorphism between $P$ and $P'$.

We observe that $\phi$ must map $I_j$ to $I_j'$ because $\phi$ preserves sources and sinks. We also note that $\phi$ must preserve familial relationships. We can therefore define the graph isomorphism $\varphi : V_1 \cup V_2 \rightarrow V_1' \cup V_2'$ as follows.
\begin{itemize}
\item For $u \in V_1$, set $\varphi(u) = u'$, where $u' \in V_1'$ is such that $\phi(u^s) = u'^s$ for $s \in \{m,f\}$.
\item For $v \in V_2$, let $U \subset V_1$ be the neighbors of $v$. Because $U$ is non-empty, there exists by construction a set of vertices in $I_2$ corresponding to $v$, all of whom mate with each other and no one else. Because it is a pedigree isomorphism, $\phi$ sends this set to a set in $I_2'$ all of whom mate with each other and no one else, and which thus corresponds to a vertex $v' \in V_2'$. We set $\varphi(v) = v'$.
\end{itemize}

We now show that if $(u,v)$ is an edge in $G$, then $(\varphi(u),\varphi(v))$ is an edge in $G'$, making $\varphi$ a graph isomorphism. Suppose $(u,v)$ is an edge in $G$. Then the vertex $(u,v)^m$ exists in $I(P)$, and there is an edge from $u^m$ to $(u,v)^m$ in $P$. Then in $P'$, there is an edge from $\phi(u^m) = \varphi(u)^m$ to $\phi((u,v)^m) = (\varphi(u),\varphi(v))^m$. But then there must be an edge from $\varphi(u)$ to $\varphi(v)$ in $G'$.
\end{proof}

The case of bipartite graphs with isolated vertices is easy to handle when checking for bipartite graph isomorphism: we ensure that there are the same number of isolated vertices in either graph, remove them, and then check for isomorphism. Therefore, Proposition~\ref{prop:GIhard}, together with the fact that bipartite graph isomorphism is at least as hard as general graph isomorphism, gives us the following theorem.
\begin{theorem}
The pedigree isomorphism problem is GI-hard.
\end{theorem}

\subsubsection{Sub-Pedigree Isomorphism is NP-Hard}
The reduction given in the previous section is easily adapted to show that sub-pedigree isomorphism is as hard as bipartite sub-graph isomorphism. Since bipartite sub-graph isomorphism is NP-hard by a trivial reduction from bipartite Hamiltonian cycle~\cite{Akiyama1980}, this gives us the following result.

\begin{theorem}
\label{thm:subpedigreehard}
The sub-pedigree isomorphism problem is NP-Hard.
\end{theorem}

Notice that Theorem~\ref{thm:subpedigreehard} already implies the following corollary about the hardness of the pedigree edit distance problem.
\begin{cor}
The pedigree edit distance problem is NP-hard.
\end{cor}
\begin{proof} By reduction from sub-pedigree isomorphism: $\mathcal{P} = (P,s)$ is a sub-pedigree of $\mathcal{P'} = (P',s')$ if and only if the edit distance between them is exactly $|E(P)| -|E(P')|$.
\end{proof}

In the next section, we will improve this result by showing that the pedigree edit distance problem is hard even to \em approximate \em in polynomial time.


\subsection{Pedigree Edit Distance is APX-Hard}
\label{subsec:editdistancehardness}

Our results about the hardness of sub-pedigree isomorphism implied that pedigree edit distance is NP-hard. We now strengthen this result, showing that pedigree edit distance is APX-hard.

\begin{remark}
There is extensive literature on the hardness/tractability of the more general problem of inexact graph matching~\cite{Conte2004, Zhang1994}. However, these hardness results do not apply to our edit distance---perhaps the hard cases of inexact graph matching are non-pedigrees.
\end{remark}

Our reduction will actually establish the hardness of the following problem:

\paragraph{The Minimum Cut/Paste Distance between Trees (MCPDT) Problem:} Given two directed rooted
unlabeled trees $T_1,T_2$, and a natural number $k$, can $T_1$ be converted into
$T_2$ using $k$ edge additions/deletions?

Showing that MCPDT is hard suffices to establish the hardness of the pedigree edit-distance problem because an arbitrary unlabeled tree can be trivially turned into an unlabeled monogamous pedigree: consider all nodes of the tree to be female and add a founding male mate to each non-leaf in the tree. This transformation doubles the cut/paste distance between trees because it exactly doubles the number of edges in each tree.

\begin{remark}
Notice that the cut/paste distance between trees is different from the subtree-prune-and-regraft (rSPR) operation for binary phylogenetic trees, since rSPR involves maintaining the binary property of a phylogenetic tree~\cite{Wu2009} and the leaves of a phylogenetic tree are labeled.  (In contrast, here we have unlabeled trees that are not binary.)
\end{remark}

To show that MCPDT is hard, we reduce from Minimum Common Integer Partition.

A \em partition \em of a positive integer $n$ is a multiset of positive integers that add up to exactly $n$. For example, $\{3,2,2,1\}$ is a partition of $8$. A \em partition of a multiset \em $S$ is a multiset union of partitions of integers in $S$. A multiset $X$ is a common partition of two multisets $S_1, S_2$ if it is an integer partition of both $S_1$ and $S_2$. For example, given $S_1=\{8,5\},S_2=\{9,4\}$, $X=\{5,3,2,2,1\}$ is a common partition of $S_1,S_2$.

\paragraph{The Minimum Common Integer Partition (MCIP) Problem:} Given two multisets of integers $S_1,S_2$ find the common integer partition of $S_1$ and $S_2$ with minimum cardinality.

Our result will rely on the following fact, proven in~\cite{CLLJ08}.
\begin{fact}
MCIP is APX-hard.
\end{fact}

We now reduce MCIP (Minimum Common Integer Partition) to MCPDT (Minimum
Cut/Paste Distance between Trees) with an $L$-reduction~\cite{Papadimitriou1988}.

\paragraph{Reduction:} Given $S_1=\{n_1,n_2,...,n_p\}, S_2=\{m_1,m_2,...,m_q\}$, we construct two trees $T_1,T_2$ with roots $r_1,r_2$ such that for each $1\leq i\leq p$ (resp. $1\leq j\leq q$) there is a distinct path from $r_1$ (resp. $r_2$) of length $n_i$ (resp. $m_j$). For example see Figure~\ref{fig1}. This means that we need to cut $T_1,T_2$ into a common forest in which each tree except the ones containing $r_1,r_2$, is a path.

\begin{figure}[htbp]
\centering
\includegraphics[width=0.35\textwidth]{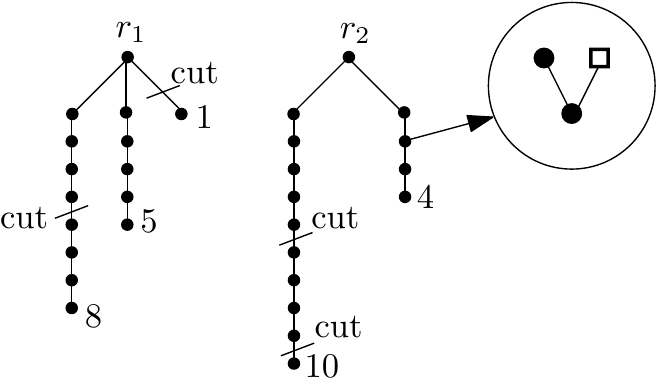}
\begin{center}
\caption{An example for the reduction.  Let $S_1 = \{8,5,1\}$ and $S_2=\{10,4\}$.  The optimal number of cut/paste operations is $2$.}
\label{fig1}
\end{center}
\end{figure}

\begin{proposition}
MCPDT is APX-hard.
\end{proposition}
\begin{proof}
We prove that the reduction given above is an $L$-reduction.  Let $\opt(\MCIP)$ and $A(\MCIP)$ be the values of the optimal and approximate solutions, respectively, of our instance of MCIP. Let $\min=\min\{p,q\}$. Then the value of the optimal solution for our instance of MCPDT is $\opt=\opt(\MCIP)-\min$. In other words, we can conclude that $S_1,S_2$ has common integer partition of size $k$ if and only if $T_1,T_2$ each can be cut into a common forest with $k-\min$ cuts.

Moreover, $A(\MCPDT)=A(\MCIP)-\min$ is the value of a feasible solution of MCPDT. Clearly we have
\begin{enumerate}
\item $\opt\leq \alpha\cdot \opt(\MCIP)$, by setting $\alpha=1$.
\item $|\opt(\MCIP)-A(\MCIP)|\leq\beta\cdot|\opt-A(\MCPDT)|$, by setting $\beta=1$.
\end{enumerate}
To see the second claim, notice that
\begin{align*}
|\opt-A(\MCPDT)| 	&= |\opt-(A(\MCIP)-\min)| \\
					&= |\opt(\MCIP)-A(\MCIP)| \\
\end{align*}
Therefore, this reduction is an $L$-reduction. As MCIP is APX-hard, MCPDT is also APX-hard. 
\end{proof}

This result implies that unless P=NP, there is no PTAS for MCPDT, and so the pedigree edit distance problem is APX-hard in general.
\begin{theorem}
The pedigree edit distance problem is APX-hard.
\end{theorem}

Certainly for leaf-labeled trees, it is well known that there is a polynomial-time algorithm for computing cut/paste distance (i.e.~on trees).  However, this algorithm does not work for leaf-labeled pedigree graphs.  Next we establish the hardness of the leaf-labeled edit distance.


\subsection{Compatibly Leaf-Labeled Pedigree Edit Distance is NP-Hard}

To prove the hardness of the compatibly leaf-labeled edit distance problem, we will take an instance of the edit distance problem without leaf labels (i.e.~$X=\emptyset$), and create an instance of the compatibly leaf-labeled edit distance problem whose solution allows us to compute the edit distance for the original edit distance instance.

\begin{theorem}
The leaf-labeled pedigree edit distance problem is NP-hard.
\end{theorem}
\begin{proof} Given non-labeled pedigrees ${\cal P}=(P,s)$ and ${\cal P}'=(P',s')$, we define compatibly leaf-labeled pedigrees ${\cal Q}=(Q,X,t,\ell)$ and ${\cal Q'}=(Q',X',t',\ell')$ as follows. ${\cal Q}$ is obtained from ${\cal P}$ by adding, for each individual $u \in I(P)$, an individual $u'$ of the opposite gender, and, for each individual $v \in I(P')$, a new individual $i_{u,v}$ which is the child of $u$ and $u'$. ${\cal Q}'$ is obtained from ${\cal P}'$ similarly: for each individual $v \in I(P')$ and $u \in I(P)$, we create an individual $j_{v,u}$ in $I(Q')$ which is the child of $v$ and $v'$, where $v'$ is a founder individual of the opposite gender as $v$, also added to $I(Q')$. Now ${\cal Q}$ and ${\cal Q}'$ have leaf sets $\{i_{u,v}\}$ and $\{j_{v,u}\}$, respectively. Let $\ell$ be defined arbitrarily on $\{i_{u,v}\}$, and let $\ell'(j_{v,u}) = \ell(i_{u,v})$. Then ${\cal Q}$ and ${\cal Q}'$ are compatibly leaf-labeled pedigrees. The following proposition completes the proof of the theorem:

\begin{proposition} 
$D_{\mathcal{P},\mathcal{P'}} = D_{\mathcal{Q},\mathcal{Q'}} - 2 \left( |I(P)||I(P')| - \min\{|I(P)|,|I(P')|\} \right)$
\end{proposition}
\begin{proof}

Without loss of generality, we assume that $|I(P)| \leq |I(P')|$. Suppose that $D_{\mathcal{P},\mathcal{P'}} = d$. Then there is a matching $M$ that achieves this distance and we may assume that the $M$ is defined on all of $I(P)$, because if not we can extend $M$ arbitrarily to all of $I(P)$ without changing the match distance.

$M$ extends uniquely to a matching $N$ defined on all of $I(Q)$ that respects the labels of the added leaf individuals and such that $N(u') = M(u)'$. We now note that, for every individual on which $M$ is defined, $N$ will have exactly one additional well-matched edge. Therefore, $|W_N| = |W_M| + |I(P)|$. We also have $|E(Q)| = |E(P)| + |I(P)||I(P')|$ and $E(Q') = E(P') + |I(P)||I(P')|$. This gives that $D_{\mathcal{Q},\mathcal{Q'}} \leq d + 2|I(P)||I(P')| - 2|I(P)|$ by Definitions~\ref{def:matchdistance} and~\ref{def:editdistance}.

Now suppose that $D_{\mathcal{Q},\mathcal{Q'}} = d$. This means that there is a matching $N$, defined again without loss of generality on all of $I(Q)$, that achieves this edit distance. If $N$ does not take every $u' \in I(Q)$ to $N(u)' \in I(Q')$, we can modify it so that this is the case, since this can only increase the number of well-matched edges of $N$. Once this is established, the same argument used above shows that, if we define $M$ to be the restriction of $N$ to $I(P)$, then $D_{\mathcal{P},\mathcal{P'}} \leq d - 2|I(P)||I(P')| + 2|I(P)|$. This completes the proof of the proposition.
\end{proof}
\end{proof}


\section{Simulation Results}
\label{sec:simulations}
We evaluated the general randomized heuristic (Section~\ref{subsec:randomheuristic}) and the dynamic programming algorithm (Section~\ref{subsec:DP}) against a brute-force branch-and-bound algorithm that finds the correct edit distance on general pedigrees in time exponential in pedigree
size. We used the modified dynamic programming algorithm for the case where an upper bound on the distance is unknown: we chose a reasonable upper bound, and if the algorithm failed to find a matching at a given step (because the distance between the two pedigrees was too large), it ran the randomized heuristic instead. Our simulations were performed on small, three-generation pedigrees so that the edit distance could be computed using the exponential-time branch-and-bound algorithm.

From the simulations, it appears that the heuristic algorithm provides a reasonable estimate of the edit distance, especially when the two pedigrees being compared are very similar to each other. The DP algorithm provides the correct answer when the two pedigrees are similar, a reasonably close answer when the pedigrees are not very similar, and an answer that matches the heuristic algorithm when the the pedigrees are very different. Of course, these results depend on the parameter $k$ we chose.


\paragraph{The simulation.}
We first drew a leaf-labeled pedigree ${\cal P} = (P, s, X, \ell)$ from a Wright-Fisher simulation where every generation has a fixed number $2m$ of individuals, there is no inter-generational mating, each monogamous couple has a number of offspring drawn from a Poisson distribution with mean $\lambda$, and all leaves are labeled. We then randomly perturbed $\mathcal{P}$ to obtain $\mathcal{P'}$ by having some fraction $x$ of non-founders choose a new parent of one gender uniformly and independently at random. (Results obtained using a perturbation model that preserved monogamy were similar.) Note that $\mathcal{P}$ and $\mathcal{P'}$ are always compatibly leaf-labeled.

\paragraph{Algorithms compared.}
We recorded the following measures of similarity for the pedigrees $\mathcal{P}$ and $\mathcal{P'}$.
\begin{enumerate}
\item Simulated Edit Path Length: $x$
\item Random-Matching Heuristic Estimate: $\hat{D}_{\mathcal{P},\mathcal{P'}} / (|E(P)| + |E(P')|)$, where $\hat{D}_{\mathcal{P},\mathcal{P'}}$ is the output of the random-matching heuristic.
\item Normalized Edit Distance: $D_{\mathcal{P},\mathcal{P'}} / (|E(P)| + |E(P')|)$, where $D_{\mathcal{P},\mathcal{P'}}$ is the output of the branch-and-bound algorithm.
\item DP Estimate: $\tilde{D}_{\mathcal{P},\mathcal{P'}} / (|E(P)| + |E(P')|)$, where $\tilde{D}_{\mathcal{P},\mathcal{P'}}$ is the output of the dynamic programming algorithm, modified for the case where there is no guarantee on distance.
\end{enumerate}

\begin{remark}
Notice that $x$ is often larger than the edit distance because the edit path taken in the simulation was longer than the shortest edit path.
\end{remark}

\begin{remark}
Since our pedigrees were randomly generated and perturbed, in practice we could not ensure the DP algorithm's condition that the pedigrees being compared be sufficiently similar at each generation. Therefore we modified the algorithm to assume a reasonable upper bound $k = 8$ on this distance and give the output of the random heuristic if no matching was found that met this condition.
\end{remark}


\paragraph{Simulation results.}
The three different results we recorded are plotted in Figure~\ref{fig:accuracy} against $x$,
the fraction of pedigree edges changed during simulation. Figure~\ref{fig:accuracy} also shows the difference between the random-matching and true edit distances. Figure~\ref{fig:runningtime} shows the running times for the three algorithms.

We see the random-matching heuristic performs reasonably well, both in terms of accuracy and time. The DP algorithm agrees with the true edit distance when that distance is small and  agrees with the random-matching estimate when the distance is large.  
However, in the intermediate area, we see that the DP produces an answer between the optimal and the heuristic value, 
because there are matchings that satisfy the distance threshold at every generation which are not the optimal matching 
and the optimal matching contains a generation that fails the distance threshold.  
Due to its accuracy, we recommend that the DP algorithm and the randomized heuristic be used together.  If degree of accuracy is not needed, then we recommend using linear-time leaf-labeled isomorphism algorithm.

\begin{figure}
h\begin{center}
\includegraphics[scale=0.35]{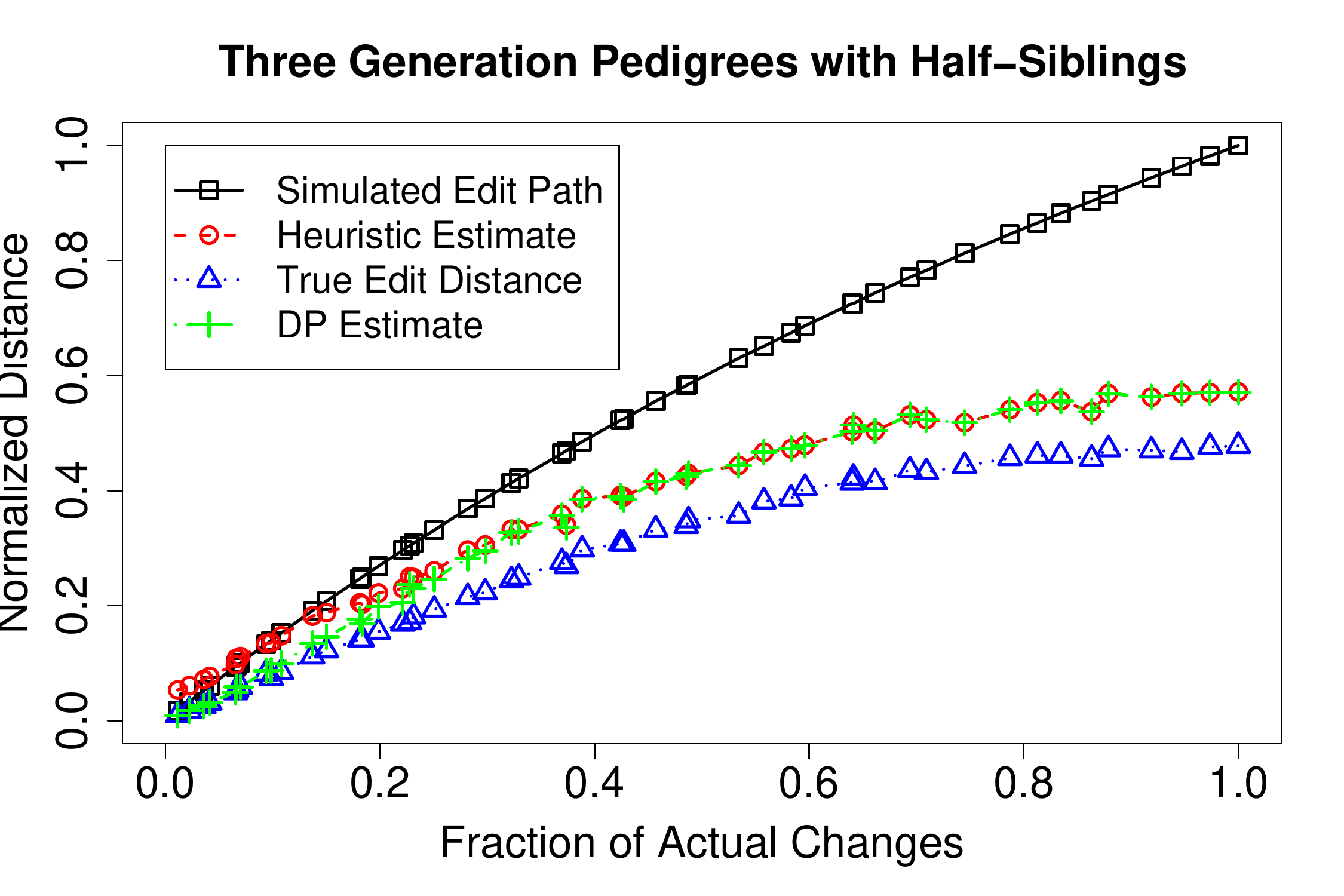}
\includegraphics[scale=0.35]{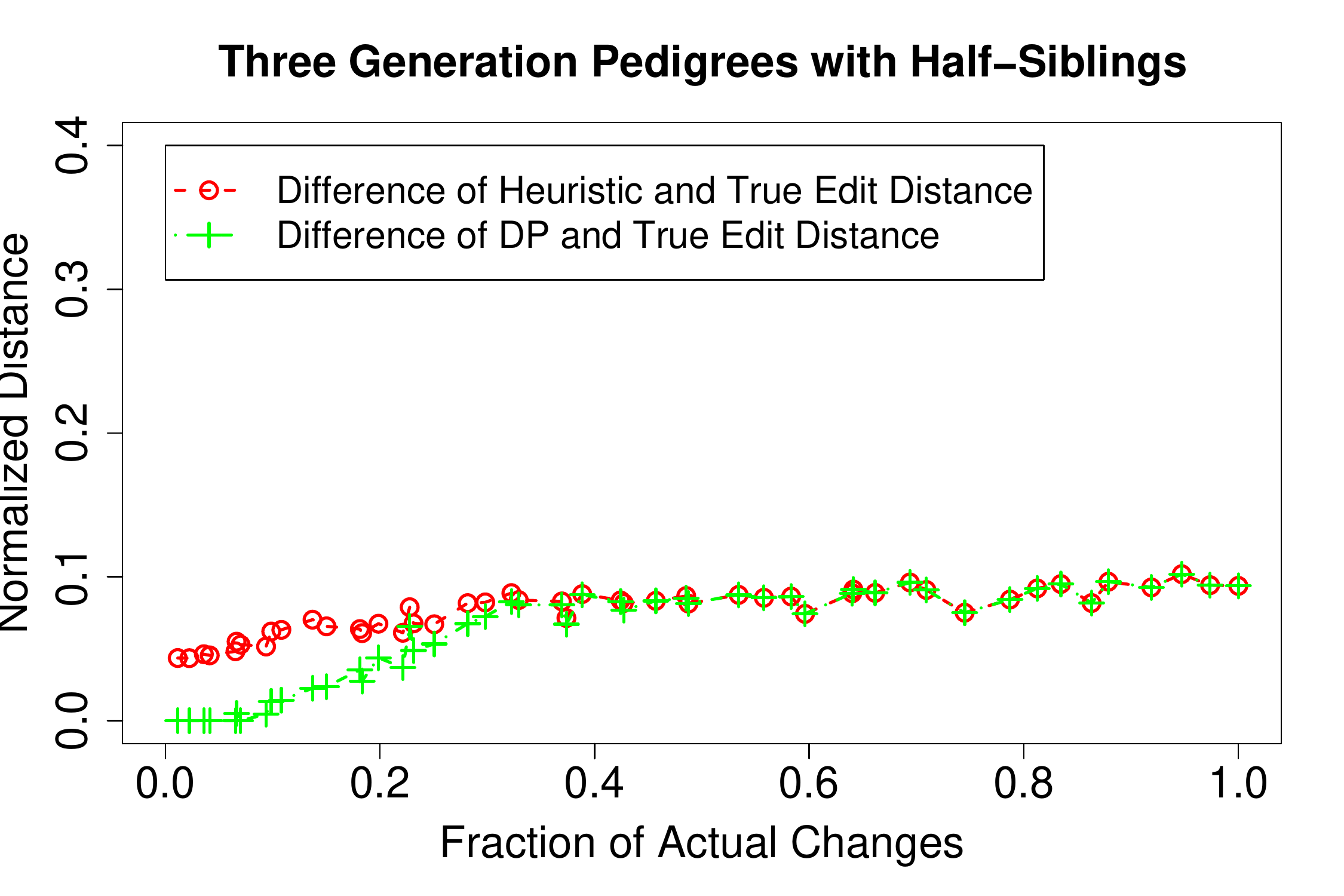}
\caption{{\bf Comparing Different Distances Estimates.} 
With $2m = 14$ and $\lambda=3$, there were 2500 pairs of pedigrees
simulated.  Each point is an average of 50 simulations.  The values of
$n$ and $\lambda$ were chosen such that the branch-and-bound algorithm
would finish computing the true edit distance.  The random matching
heuristic yields an estimated edit distance which is fairly close to
the true edit distance.  The DP algorithm performs nearly perfectly
for small numbers of actual changes, while it returns the solution
found by the random-matching heuristic when it cannot find a solution
for parameter $k=8$.  The left panel shows the accuracies of each
algorithm.  The right panel shows the difference in accuracy between
the true edit distance and each distances returned by the
random-matching heuristic and DP algorithm.  }
\label{fig:accuracy}
\end{center}
\end{figure}

\begin{figure}
\begin{center}
\includegraphics[scale=0.3]{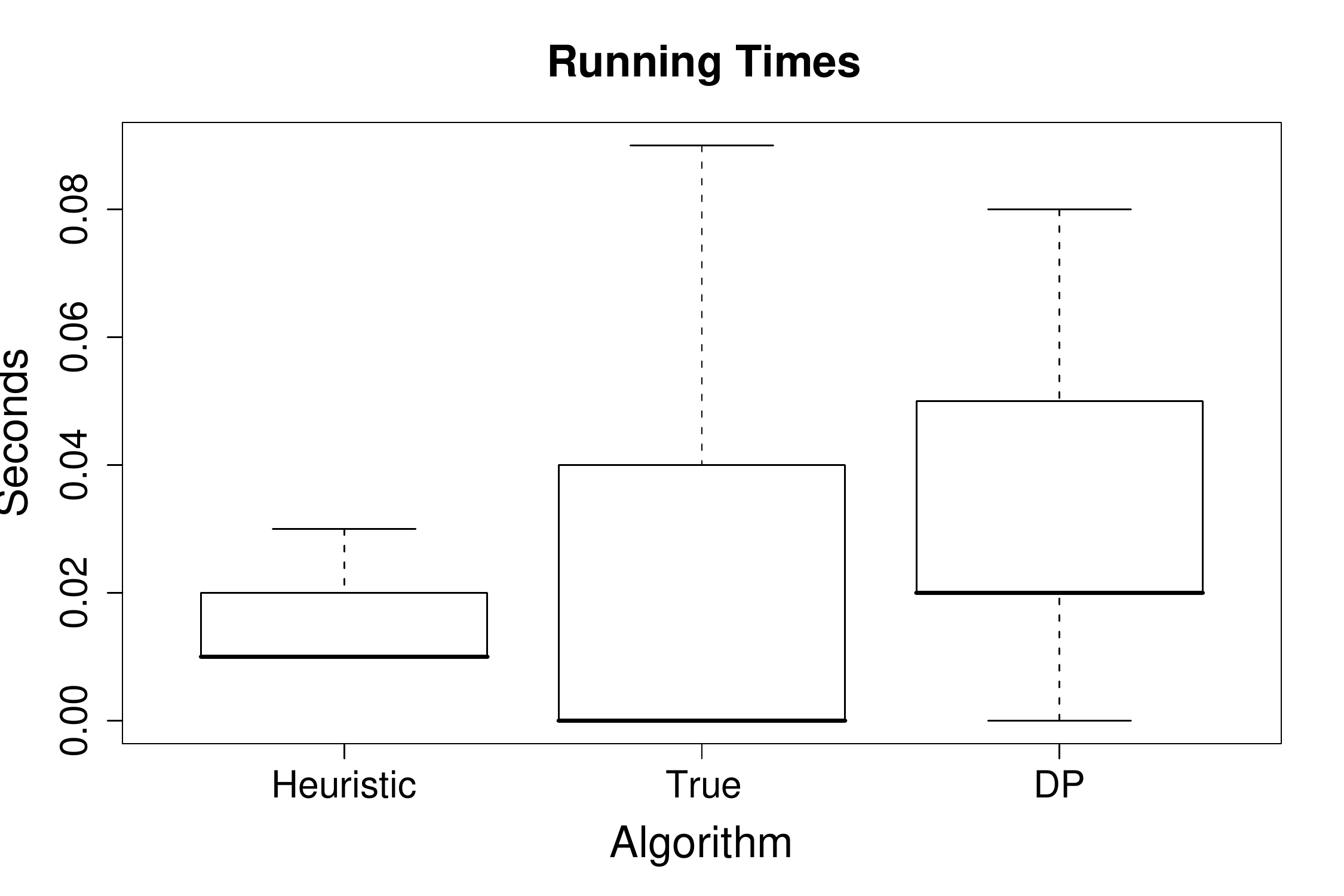}
\caption{{\bf Running Times.} 
These are box plots comparing the running times of the three different
algorithms: heuristic algorithm, branch-and-bound algorithm, and the
DP algorithm.  The heavy line is the median, the rectangle indicates
the first and third quartiles.  In this case the median is coincident
with the first quartile for all three algorithms.  The outliers are
not shown; specifically, there are a number of very long execution
times for the optimal algorithm.  }
\label{fig:runningtime}
\end{center}
\end{figure}


\section{Discussion}
\label{sec:discussion}

In this paper, we introduced two pedigree comparison problems---pedigree isomorphism and pedigree edit distance---and we presented algorithms and hardness results for both. Several interesting open questions remain:

\begin{itemize}

\item {\bf Fractional edit distance.} An alternative definition of
edit distance could be based on fractional matchings: instead of minimizing
over all one-to-one matchings of vertices in the two pedigrees, we could
allow a vertex of $P$ to be matched to multiple vertices in $P'$ with
weights summing to one. Such a distance could be easier to compute,
although the biological interpretation is less clear. It would also be
interesting to explore the relationship between the definition
presented in this paper and this alternate definition.

\item {\bf Pedigrees with inter-generational mating} Another open problem of interest is how the edit distance algorithms work on pedigrees with inter-generational mating.
The simulations used here were based on the Wright-Fisher model and
did not allow any inter-generational mating events.  It may be of
interest to simulate the pedigrees using a birth-death model such as the
Moran model where inter-generational mating is allowed. Such a
simulation would allow the evaluation of the distance heuristics on
non-regular pedigrees.

\item{\bf Pedigree isomorphism without labels} A very interesting open problem is that of pedigree isomorphism without labels.  Since the graph isomorphism problem is reducible to it, it is conceivable that existing algorithms for graph isomorphism might be of use for pedigree graphs.
\end{itemize}

Comparison of pedigrees is an interesting and important problem. Here, we have taken the first steps towards understanding and solving it.


\section*{Acknowledgments.}
Many thanks go to Yun Song for stimulating discussions and to Eran
Halperin for the random pedigree simulation.

\section*{Author Disclosure Statement.}
No competing financial interests exist.

\bibliography{pedigree}
\bibliographystyle{plain}

\end{document}